%% file: main.tex
\documentclass[conference]{IEEEtran}

\usepackage[letterpaper,
            left=1.02in,
            right=1.02in,
            top=0.75in,
            bottom=1.04in,
            footskip=.25in]{geometry}

\IEEEoverridecommandlockouts

\usepackage{cite}
\usepackage{amsmath,amssymb,amsfonts}
\usepackage{algorithmic}
\usepackage{graphicx}
\usepackage{textcomp}
\usepackage{xcolor}
\def\BibTeX{{\rm B\kern-.05em{\sc i\kern-.025em b}\kern-.08emT\kern-.1667em\lower.7ex\hbox{E}\kern-.125emX}}

\usepackage{enumitem}

\usepackage{graphicx}
\usepackage{algorithm2e}
\RestyleAlgo{ruled}
\usepackage{amsthm}
\usepackage{url}
\usepackage{subcaption}

\newcommand{\ignore}[1]{}
\newcommand{\shir}[1]{\textit{\textcolor{blue}{[Shir]: #1}}} 
 
\newcommand{\red}[1]{\textcolor{red}{#1}}

\newcommand{\sysName}{NODE }
\newcommand{\sysNameNoSp}{NODE}

\title{NODE: Network Wide Top-K Flows in the Data Plane}

\author{\IEEEauthorblockN{Eitan Stein}
\IEEEauthorblockA{\textit{The Open University of Israel} \\
Israel }
\and
\IEEEauthorblockN{Lior Zeno}
\IEEEauthorblockA{\textit{Technion} \\
Israel }
\and
\IEEEauthorblockN{Shir Landau Feibish}
\IEEEauthorblockA{\textit{University of Haifa} \\
Israel }
}

\usepackage{graphicx}
\usepackage{algorithm2e}
\RestyleAlgo{ruled}
\usepackage{amsthm}
\usepackage{url}
\usepackage{subcaption}

\newtheorem{theorem}{Theorem}[section]
\newtheorem{lemma}[theorem]{Lemma}

\renewcommand{\thesection}{\arabic{section}}

\begin{document}

\maketitle
\input{AbstractSection}

\input{Intro}
\input{Background}

\input{Challenges}

\input{Algorithm}

\input{networkwide}
\input{Evaluation}
\input{Conclusion}
\input{Acknowledgments}

\bibliographystyle{plain}
\bibliography{references}

\end{document}

%% file: AbstractSection.tex
\begin{abstract}
    Monitoring network traffic is crucial  for most network tasks, such as,   identifying and blocking attacks, pinpointing failures and engineering and rerouting heavy traffic to maintain high throughput. One important metric when monitoring the traffic is finding the top-k heavy flows, that is the k \emph{heaviest} flows in the traffic.
    Programmable networks allow performing advanced network analysis right in the data plane. In recent years, various solutions have been proposed for efficiently finding the top-k heavy flows within a single switch.  However, at times we may need to find the \emph{global} top-k flows. 
    Existing solutions for global top-k detection use a centralized controller that collects and aggregates the measurements performed in each of the switches.  
    Yet, the process of sending information to the control plane 
    and then having the controller send back the information to the switches can be very lengthy. 
    In order to be able to detect and mitigate short-lived events, 
    solutions that work completely within the data plane are needed. 
   In this paper we present \sysNameNoSp, a network-wide top-k detection algorithm that operates exclusively in the data plane. \sysName allows the switches to aggregate information from all other switches in the network, and ensures that eventually all switches hold an \emph{identical} global top-k table. We show that \sysName manages to detect global top-k flows on both synthetic and real traces, with a recall rate of over 95\% while using less than 300KB per switch. 
\end{abstract}

%% file: Intro.tex
\section{Introduction}

\ignore{
\shir{
This is the basic structure for intro (more here:
}
\shir{
https://www.cs.columbia.edu/~hgs/etc/intro-style.html 
\begin{itemize}
    \item para. 1: motivation: broadly, what is problem area, why important?
    \item para. 2: narrow down: what is problem you specifically consider
    \item para. 3: “In the paper, we ….”: most crucial paragraph, tell your elevator pitch
    \item para. 4: how different/better/relates to other work
    \item para. 5: “The remainder of this paper is structured as follows”
\end{itemize}
}
}

Networks require real-time monitoring for a wide range of uses such as traffic engineering and identifying failures and threats.  \cite{MicroTE_traffic_engineering, 
robustDistribMonitorTraffic, bellIOTSecurity, trafficEngInSDN, TCPincast, netwideRoutingObliviousHH}. 
One of the common monitoring tasks is to identify
the top-k flows. This can be useful in routing decisions, for example, as improved routing of heavy flows may alleviate congestion and increase network throughput \cite{robustDistribMonitorTraffic, loadSensitiveRouting}
. Identifying the top-k flows can also assist in detecting and mitigating volumetric attacks such as DDoS \cite{jaqen}. 
Furthermore, with the advancement of programmable networks, many solutions have been developed for performing network monitoring right inside the data plane. Finding top-k flows in a programmable switch is challenging due to the harsh constraints on memory and processing,
yet, there are existing solutions for finding top-k flows \cite{PrecisionJournal} that can be deployed on a programmable switch. However, these solutions are only suitable for a single switch, whereas some tasks require finding the \emph{network-wide} top-k flows. 

In existing solutions, in order to find network-wide top-k flows, each switch processes the traffic that it sees, and then additional processing needs to be performed \emph{collectively}. Some solutions use samples that are either collected in the switch~\cite{Aroma} or sent to the collector individually~\cite{netflow,sflow}. Due to the high overhead incurred by collecting these samples, often only a limited amount of packets are sampled (e.g. 1 in 30K~\cite{facebookSampling}), which reduces the accuracy of the measurement.
In other solutions, in order to avoid sampling, some processing is performed in the data plane to identify the locally heavy flows~\cite{Carpe, adaptiveThresholds, IncrementalDeployment, flowradar}. The processed information is then sent to a centralized controller that aggregates the data, and then passes the information of the network-wide top-k flows back to the switches as needed. This too is not optimal, as interaction with the controller may take a long time and thus is not suitable for time-sensitive measurements, such as short lived bursts. Furthermore, having each switch calculate the top-k flows locally can miss network wide heavy flows that spread across the network and are `small' in some or all individual switches. Looking for such smaller flows locally will incur significant overhead in either switch resources or communication or both.

In order to find top-k flows only within the data plane, switches need a mechanism for sharing information. 
Swish~\cite{swish} is a shared state management system in the data plane, it can be used to maintain 
a distributed Count-Min Sketch (CMS)~\cite{CMS}. Each switch maintains a local CMS, and the Swish framework distributes the information from each local CMS to all other switches. Each switch then merges the sketches locally to get a combined network-wide CMS. In this manner 
Swish enables 
finding network-wide heavy hitter flows completely in the data plane. 
Furthermore, Swish\cite{swish} achieves better speed by an order of magnitude compared to using a centralized controller to obtain the combined sketches.
However, this mechanism is not suited for finding \emph{top-k} flows. CMS only maintains counters, and not flow IDs. Furthermore, flow IDs are required to extract the flow count estimation from the sketch. 
Thus a different solution is needed for finding network-wide top-k flows.   

 \textbf{The \sysName framework.} 
We present \sysName (Network wide tOp-k in the Data planE) a system for finding network-wide top-k flows completely in the data plane. 
\sysName maintains the local top-k flows in each switch and then both distributes and aggregates them in the data plane, without controller assistance. 
 %
 More specifically, in order to find the network-wide top-k flows, \sysName performs three main tasks: 1) \sysName creates a local top-k table in each switch independently; 2) \sysName builds on Swish to share the local top-k tables from each of the switches among all the other switches, in order to find the global counters of \emph{local} flows. 3) In each switch, \sysName  merges the information found in \emph{all} of the local top-k tables, to eventually form a global top-k table in every switch in the network;  
\sysName performs the entire process completely within the data plane and does not require a central controller or collector to aggregate the information.
%
%
Furthermore, we implement \sysName on 
a simulated testbed and in P4 code for the Intel Tofino Wedge-100 programmable switch~\cite{tofino}, and evaluate the effect of various parameters on \sysNameNoSp's performance, as well as compare it to controller based solutions.

In the following, we provide the background
\S~\ref{sec:background}, 
in \S~\ref{sec:challenges} we list the challenges of merging network-wide data in the data plane, and \S\ref{sec:system} describes the \sysName framework. 
Finally, we show results of \sysNameNoSp's evaluation (\S~\ref{sec:evaluation}) and finish with a conclusion (\S~\ref{sec:conclusion}).

%% file: Background.tex
\section{Background and Related Work}\label{sec:background}

We provide the background needed 
for understanding \sysNameNoSp's design and functionality.
We 
describe the processing restrictions of the data plane, and give a short overview of the existing solutions for finding top-k flows both in and out of the data plane, as well as methods for information sharing 
in the data plane.

\subsection{Data Plane Restrictions}
For network hardware to process packets at line rate (i.e., Tbps speed), they must impose harsh restrictions on both the memory resources and the computation done in packet processing. Programmable switches (such as PISA~\cite{RMT}) 
have a 
 feed-forward packet processing pipeline comprised of a sequence of a small number of stages. Each stage has its own limited amount of unique memory, to enable stateful computations. Access to this memory (read/write operation) must be done when the packet is being processed by each stage. 
 For example: a packet cannot access the memory of stage $i$ while being processed in a different stage $j$. Additionally, the amount of memory that can be accessed in any stage while processing the packet is very limited. 
 If additional processing on the packet is required, the packet may be recirculated to the beginning of the pipeline, however performing too many recirculations can affect throughput as it requires the switch to processes the packet again.

\subsection{Finding Top-K Flows in the Data Plane} 
There have been various solutions proposed for finding top-k flows and heavy hitters.  
We describe some of these works as they will help to understand the choices made in designing \sysNameNoSp.

\textbf{Space Saving. } Space Saving~\cite{SpaceSaving} is a well known technique that maintains a subset of the items in the stream and a counter for each member of the subset. When a new packet with ID $x$ is processed, if $x$ is in the subset, its associated counter is incremented by $1$. If $x$ is not saved in the subset, the algorithm finds the member with the \emph{smallest} counter (i.e. $MinCount$) in the subset and replaces that member with $x$ and increments its counter by $1$. Although Space Saving succeeds in finding heavy flows (if their frequency is high enough compared to the size of the stream), it may over-estimate the frequency. This is especially significant in small flows that may be given a very high frequency estimation.  

\textbf{RAP. } Random Admission Policy ~\cite{RAP} expands on the same idea but, upon seeing an ID that is not in the subset, instead of \emph{always} replacing the entry with the lowest counter, it will only replace it with a certain probability (the exact probability is \(1/({MinCount} + 1)\)). This cancels some of the noise (excess counter increases) by low frequency flows and thus reduces the flow count overestimation.

\textbf{HashPipe. } Both Space Saving and RAP find top-k flows but they cannot work within the restrictions of the data plane. The main obstacle is obtaining the smallest counter. In order to find this counter, we must sort the items or access all of the counters to find the minimum.  This would require many more memory accesses than are available. Furthermore, after finding the minimum counter we would need to access the proper stage where it is stored to change the corresponding ID.  
HashPipe ~\cite{HashPipe} was the first solution for finding top-k flows in the data plane. 
To avoid looking for the minimum across all counters, 
HashPipe uses a table that is divided into $d$ separate vectors, each with its own hash table, placed in separate stages, which are used to create a \emph{rolling minimum}. 
For each vector, the packet will be hashed into a single location based on the packet ID. If the ID in the vector matches the packet's ID, the counter in the vector is incremented. Each packet will also maintain an additional ID and counter in packet metadata, which will be used to maintain the minimum. That is, if the IDs in the packet and vector \emph{don't} match, the counters are compared. If the counter in the packet is larger than the counter in the vector, the algorithm replaces the values in the vector with the ID and counter carried by the  packet and continues processing the packet with the ID of the smaller flow and its counter, so that the minimum values will eventually 'roll' out of the table.



HashPipe's algorithm suffers from one main issue. 
Precision~\cite{PrecisionJournal, Precision2018} shows that HashPipe does not meet the requirements to run on programmable switches. Since it requires both accessing the ID before the counter to decide whether to increment the counter, but also requires to access counters before IDs to compare the counters to decide whether to switch the packet ID and the saved ID.

\textbf{Precision. } In order to create a top-k algorithm that can run on programmable switches, Precision~\cite{PrecisionJournal, Precision2018} builds on both RAP and HashPipe by using a similar structure to HashPipe's table while handling packets similarly to RAP. When each packet reaches the hash table in the $j$-th vector, it is hashed (just like in HashPipe) to check if its ID matches the ID in the vector. 
If the IDs match, the associated counter is incremented. If a match is not found in any of the vectors it will recirculate to replace the smallest observed counter with probability \(1/({MinCount} + 1)\). This way there are no excessive recirculations. Note that Precision can run on programmable switches since it always first checks the ID and only then handles the counter.

\textbf{Controller based network-wide top-k.} Several solutions have been proposed for finding network-wide top-k flows~\cite{IncrementalDeployment, adaptiveThresholds, Carpe, Aroma}. Some of these solutions include data plane analytics in programmable switch networks which can be used for finding top-k flows, 
however these solutions depend on either pulling or pushing data to a centralized controller which builds the network-wide top-k list and sends this list to all switches. We wish to find network-wide top-k completely in the data plane in order to create the table faster as seen in the comparison between Swish and a centralized controller \cite{swish}.  
%
FlowRadar~\cite{flowradar} is a known controller based algorithm for finding network-wide heavy hitters. 
FlowRadar maintains a Bloom filter \cite{bloom1970} and an array of encoded flows.  When a new packet arrives it checks the bloom filter to see if it is a new flow or not. If the flow is new, it encodes the flow ID into several array locations via hashing. The encoding result is a XOR of the existing value in the array with the packet's flow ID. Then it increments the counters saved in those slots. If the flow is not new, it only increments the counters in the hashed cells. 
FlowRadar sends this array frequently to the controller which decodes the flows and sends back network-wide information. Due to FlowRadar's implementation and frequent controller updates, it is likely to find all flows, and especially heavy flows as long as it has enough memory.

\subsection{Information Sharing in the Data Plane. }

There are solutions that require sharing data between switches in the network. Most solutions either share only limited amount of data~\cite{netchain}, or use the control plane for assistance. That is, they communicate information from switches to the controller which processes it and sends relevant information to other switches~\cite{Carpe, flowradar, adaptiveThresholds, MVSketch, IncrementalDeployment}. 
Sharing large amounts of data between switches \emph{completely} in the data plane is more challenging as it requires both communicating the information in the face of errors and processing the information with the limited resources of the data plane.

\textbf{Using Swish to exchange data between switches in the data plane.} Swish~\cite{swish} is a mechanism that allows managing shared state across different switches completely in the data plane. It allows replicating and sending data between switches without having to rely on a central controller.
Swish guarantees that all packets will be delivered successfully without duplicates. Swish can also handle failures, and in order to do so, it must send the data from a static source, so it can re-send the information if necessary.

\subsection{Additional Related Work. }\label{sec:relatedWork}
We mention a few works on network-wide top-k or network-wide heavy hitters, all depending on controller interaction. 
%
Some solutions send only partial information to the controller, either by sampling or by decisions made via metrics such as 
thresholds.
Carpe~\cite{Carpe} combines probabilistic counting on the switches with probabilistic reporting to the central coordinator, guaranteeing that communication costs do not grow proportionally with the number of switches.
An earlier work by Harrison et al.~\cite{adaptiveThresholds}, reduces communication overhead with the controller, by  using local and global per key thresholds to limit reporting. The MV-Sketch~\cite{MVSketch} is a data structure which saves heavy flows candidates and merges the sketches from all switches in the central controller. 
AROMA \cite{Aroma} suggests a method to effectively sample packets and flows using programmable switches using controller analysis while taking into account that packets may potentially go through multiple switches.

%% file: Challenges.tex
\section{Challenges}\label{sec:challenges}

In order to create a global top-k detection algorithm in the data plane, each switch in the network needs to first identify the locally heavy flows, and then merge this information with the heavy flows identified by each of the other switches in the network.
However there are some significant challenges in achieving this in the data plane, which we will now describe. 

\textbf{Memory restrictions.} 
A straightforward solution for merging the network-wide information, would be to have each switch broadcast its local top-k table to all other switches. Each switch would then store all of the tables and then process them all together.
However, since we are using programmable switches, our memory and computational abilities are very limited, such that each switch cannot simultaneously  hold \emph{all} of the received data from \emph{all} the other switches and process all of the data at once. Instead, each switch needs to process the information from other switches, as it is received, in a streaming fashion.

\textbf{Split heavy flows. }
To overcome this issue, let's consider another simple approach, where switches filter smaller flows from top-k tables. Each switch computes a local top-k table, and then sends the table to other switches. When a switch receives a packet (that comes from a local table of another switch) holding a flow ID and counter, it will compare the packet's flow ID to IDs stored in its own table. If they match, it will add the packet's counter to the counter stored with that ID. If the packet's ID is not found in the table and the packet's counter is larger than the smallest counter in the table (or the smallest counter that can be found within the memory access restrictions), the stored ID and count of the smallest item in the table would be replaced with the packet's ID and counter, such that the switch will store the heavier flow of the two.

Yet, this solution is inherently flawed. A heavy flow could be split across the network into small pieces, which may be missed with this solution. When attempting to merge the global information, if the flow isn't saved in the local switch table, each small piece of the flow could be considered as a small flow and be discarded in favor of flows that appear to be heavier. For example, if we had $10$ switches each holding $5$ flows, the first switch holds flows with counters ranging from $200$ to $300$, and does not hold some flow $x$, while every switch except the first holds flow $x$ in its table with counter $100$. When the first switch receives information about flow $x$ from the other switches, each packet will hold a counter of only $100$, which is smaller than counters already saved in the first switch, so it will disregard it even though the global counter of flow $x$ is $900$ which is much heavier than the flows held in the first switch table, 
so, we may miss heavy flows completely.

\textbf{Memory access dependencies. } 
Another issue with the above solution is that, similarly to HashPipe~\cite{HashPipe}, it requires both checking the ID first to find out if the ID is in the table, and checking the counter first to check if its larger.
That is, on the one hand we wish to first compare IDs in order to see if the IDs match and then aggregate their counts, which requires accessing the IDs earlier in the pipeline, before handling the counters. But on the other hand, we might want to remove flows with a low count from the table (a packet with counter 1000 should be able to replace a saved flow with counter 100), but that requires to first compare counts and only then handle the IDs. However, if the IDs come before the counters in the pipeline, once we reach the counters the IDs can't be accessed without re-circulation.

Yet, we note that if we wish to use re-circulation, 
changes might happen to the switch table (by other updates from other switches) \emph{while} a packet recirculates.
For example if a switch holds a small flow $x$ with counter $100$, and a packet from another switch holds a flow $y$ with count $500$. The packet with flow $y$ notes $x$ as the smallest flow in the switch and decides to recirculate to replace it. While it recirculates, another packet arrives from another switch with flow $x$ and count $1000$, this packet finds a matching flow ID and increments the counter in the table to $1100$. When the packet with flow $y$ finishes recirculating and replaces $x$, it now replaces a flow with a larger counter ($x$ with $1100$ compared to $y$ with $500$), but it cannot know that without checking the counter, which will be done after already replacing the stored ID in the table and setting it to $y$. Note that Precision~\cite{PrecisionJournal} does use recirculation for similar purposes, however it does not suffer from this drawback since each increment done is by at most $1$, which means the smallest counter cannot change drastically while a packet recirculates and is likely to remain one of the smallest if not the smallest counter in the table.

%% file: Algorithm.tex
\section{The \sysName Framework}\label{sec:system}

In order to find network-wide top-k flows in the data plane while addressing these challenges, we designed \sysNameNoSp. \sysName shares the needed information between switches to create an identical global top-k table in every
switch, and performs all data sharing and processing exclusively in the data plane. 
In this section we present an overview of \sysNameNoSp, followed by a detailed description of the data structures and algorithm. 

\subsection{Overview}

In order to solve the above challenges, \sysName takes a deterministic approach and split the process into three main parts: 

\textbf{1) Creating a local top-k table.} Each switch processes the incoming packets to create a \emph{local} top-k table. 
\textbf{2) Global counter aggregation.} Each switch sends the contents of its local top-k table to all other switches. Each switch then aggregates the counters \emph{only} for flow IDs that are \emph{already found} in the \emph{local} top-k table to find their \emph{global} counters. In this way, the switch finds the global counts of the top-k flows in its \emph{own} table.  
\textbf{3) Consolidation of all tables.} Then, each switch sends the contents of its local top-k table, with the aggregated counts to all other switches. Each switch filters the flows with smaller counters, such that only the flows with the higher counters remain in the table. Once this process is completed, each switch holds a global top-k table.


Information sharing between switches is done completely in the data plane, using the Swish framework~\cite{swish}. That is, the entire process does not require any controller interaction and is performed completely within the data plane. Note that Swish guarantees that all information is sent exactly once in a single Swish pass (assuming no failures).

We address the challenges described above as follows:

    \textbf{Memory constraints. } \sysName doesn't require saving all the information from every switch. Instead, \sysName processes each packet that arrives from each switch 
    in a streaming fashion (i.e., in a single pass).

    \textbf{Handling split flows. } By performing global counter aggregation and then consolidating the tables, \sysName uses the \emph{global} information to calculate the global counter of each flow in order to identify the top-k flows. That is, when \sysName performs the consolidation round between the switches (step 3 above), each packet will hold a flow ID and its \emph{global} counter. Therefore, when \sysName consolidates the tables it considers the entire flow count and therefore it doesn't need to worry about making decisions with only partial information.

    \textbf{Order of accessing IDs and counters. } As we will describe shortly, \sysName uses different table formats for each stage of information sharing (steps 2 and 3). For global counter aggregation, the ID is placed before the counts in the pipeline, and for table consolidation, the counts are placed before the IDs. By doing so, when calculating global counters for local flows, \sysName first accesses the ID to search for matching flows and then accesses the counter to increment it if the IDs match. Afterwards, when receiving packets with global counters, \sysName uses a separate table where the counters are stored before the IDs, so it first compares counters and only then modifies the ID in the table if the stored counter is smaller than the packet's counter.
    \sysName does not use recirculation at all for packets received from other switches and therefore, does not need to address issues that arise from using re-circulation.

\subsection{Components of \sysName}
\begin{figure*}[t]
\centering
    \includegraphics[width=0.99\textwidth]{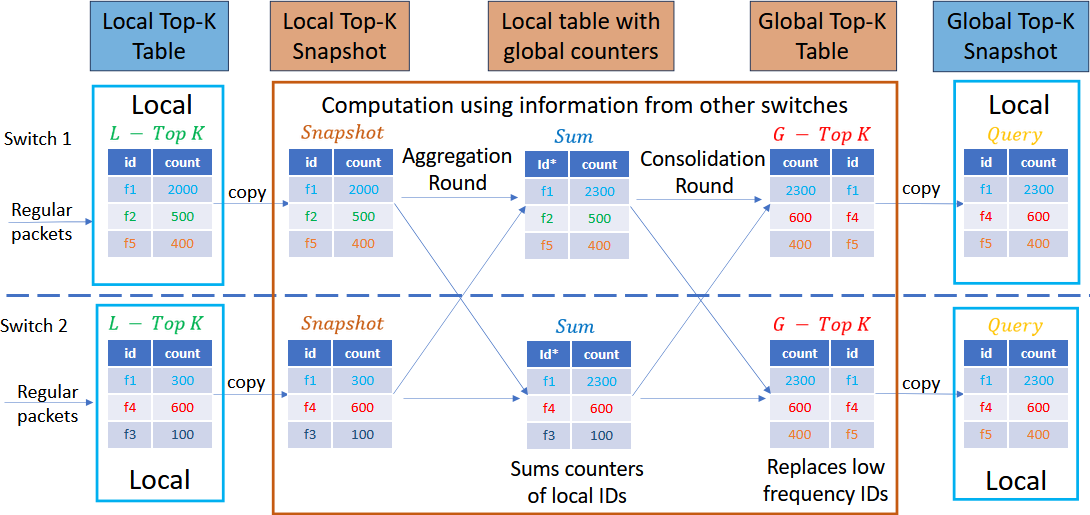}
    \caption{An example of \sysName flow. The Sum table uses the same IDs as the Snapshot table.}
    \label{fig:algorithm overview and flow}
\end{figure*}

As shown in Fig.~\ref{fig:algorithm overview and flow} \sysName makes use of five tables, that are maintained in each switch.  All of the tables use the same hash functions - that way a packet with ID $x$ will get hashed to the same table location in every table. We will now go over each of \sysNameNoSp's tables:

\noindent \textbf{L - Top K: Local top-k table. } \sysName makes use of a local top-k detection algorithm. As packets enter the switch, they are processed by that algorithm to create a local top-k table.

\noindent \textbf{Snapshot: Snapshot of local top-k table. } Once a local top-k table has been created, \sysName sends the contents of this table to all of the other switches in the network. Yet, in order to allow failure handling when sending table information to other switches,  \sysName must maintain a static source of the local top-k table that is used in the global top-k detection process.

\noindent \textbf{Sum: Local top-k table with global counters. } 
    In each switch this table maintains the \emph{global} count of the flows found in its \emph{local} top-k table. That is, this table aggregates the network-wide counts of the flows that were locally heavy and are found in the local top-k table.

\noindent \textbf{G - Top K: Global top-k. } 
    This table maintains the global top-k flows across all flows in the network. That is, using the information about top-k flows received from $Sum$ tables across the entire network, \sysName will determine which flows are globally heavy (whether they were found in the local top-k table or not) and maintains them in this table.

\noindent \textbf{Query: global top-k snapshot. } Once the global top-k flows have been identified, \sysName needs to maintain this information so that information about the global top-k flows may be used during regular packet processing. This is a static table, that is maintained until a new global top-k table is computed.

Note that \sysName takes as parameters: 1) The number of vectors (sub-tables) $d$, i.e. the number of vectors of $(ID, count)$ pairs in \emph{each} \sysName table. Each vector has a hash table, thus $d$ indicates the number of hash-tables used for each one of \sysNameNoSp's tables. Note that $d$ is identical in every \sysName table, and the hash functions used for each vector are identical in all of \sysNameNoSp's tables.  That is at any vector $i$ all \sysName tables share the same hash function. 2) The size $s$ of each vector, i.e. the number of $(ID, count)$ pairs saved in each vector. Note that $s$ is identical in all vectors and in all tables.

\subsection{Detecting Global Top-K}

We will now describe how \sysName uses the above tables to find the global top-k flows.
To better understand the process, we will follow each table in Fig.~\ref{fig:algorithm overview and flow} and how it interacts with the flows passing through it. Note that Fig.~\ref{fig:algorithm overview and flow} has a single vector per table for simplicity, in a normal setting each table will have multiple vectors, each holding pairs of IDs and counters.

\subsubsection{Creating a local top-k table}

\sysName first finds the local top-k flows. 
The Local top-k table maintains the local top-k flows of packets that traverse the switch. That is, every packet that traverses the switch is processed by the local top-k algorithm (e.g. Precision as it can run on programmable switches and creates the same table format that \sysName uses), and the flow ID is inserted into this table accordingly. This table \emph{continuously} maintains the top-k flows. 
In Fig.~\ref{fig:algorithm overview and flow}, switch $1$ has flows $f_1$,$f_2$,$f_5$ in its local top-k table, while switch $2$ has $f_1$,$f_3$,$f_4$ in its local top-k table. 

\subsubsection{Global counter aggregation}

In each switch, \sysName sends the local top-k information to all other switches. We call this the Aggregation round. 
\sysName first copies the contents of the local top-k table into another table called the {\it Snapshot} table. This table will remain \emph{static} throughout the process of finding the global top-k flows. There are two main reasons for maintaining this static table. First, \sysName cannot send the entire table at once to all of the switches. The table is sent in parts and therefore needs to remain static to avoid inconsistencies in the data sent to other switches. Second, packets can get lost on the way to the other switches. Therefore, \sysName requires a static copy to support potential information sharing failures. 
Since the local top-k table is continuously updated with every new packet that enters the switch, it cannot support these operations, and thus a second table is used to maintain information sent to other switches. 
We can see in Fig.~\ref{fig:algorithm overview and flow} that the Snapshot table in each switch is an exact copy of the local table right before the Aggregation round.

In order to sum up the global counters of local flows, \sysName (once again) cannot use the local top-k table that is being modified as new items come in, and it cannot change the snapshot since it may need to resend the information in case of failure. Therefore a third table called Sum is used. 
In Figure \ref{fig:algorithm overview and flow}, Sum table is separated from the Snapshot table for clarity, but to save space, it shares the same exact IDs of the Snapshot table (without changing them) while having its own count values for each flow ID $x$, which will eventually become the global count of $x$.

In order to aggregate the global counts of the local top-k flows,  
\({Sum}\) is initialized to be a copy of \({Snapshot}\). 
When an ID and count from another switch are received, it will check whether the received packet's ID matches an ID that is currently saved in the {\it Snapshot} table. If the ID is in the table, it will add the received counter to the counter currently held in the \({Sum}\) table. For example, Switch 1 has flow $f_1$ with a count of $2000$ in the Snapshot. Thus in the Aggregation round it will send this information to Switch 2. Since Switch 2 also has $f_1$ in the Snapshot table, it will add the count that is received (i.e. 2000) to its own count such that the \({Sum}\) table now has a count of $2300$ for $f_1$. If the ID is not in the table then it does nothing. For example, flow 
$f_4$ is only found in Switch 2. When the information about $f_4$ reaches Switch 1 in the first round, it is disregarded. 

Once all packets from the Aggregation round have been received, \({Sum}\) becomes a static table. This means we can use the table itself to send data to other switches without having to create a copy of it. However, this also means we cannot make changes to this table, as information of \({Sum}\) tables is received from other switches.
We note that at this point \({Sum}\) holds the global counter for each \emph{local} flow as we can see in Fig.~\ref{fig:algorithm overview and flow}.

\subsubsection{Consolidation of all tables.}
 
Once all counters of each of the flows have been aggregated, \sysName is ready to consolidate the tables and determine which flows are the global heavy flows. 
In each switch, \sysName sends the data from the {\it Sum} table to all other switches. We call this the Consolidation round. 

\sysName uses a fourth table for this process: \({G-TopK}\).
As can be seen in Fig.~\ref{fig:algorithm overview and flow}, this table is structured  differently from other tables as the counters are placed before the IDs. This allows \sysName to first compare the counters and only then handle the IDs.
\({G-TopK}\) starts as an empty table, it does not copy the local \({Sum}\) table values, but instead treats them as input packets as if they came from another switch. 
When processing a packet containing information from a \({Sum}\) table, it will use the flow ID to find the location it is hashed to (like it would do in other tables), then it compares the packet's counter with the saved counter in the table. If the packet's counter is larger than the counter in the table, the packet will replace the flow ID and count values with its own.
Since the counter is located before the ID in the pipeline, the packet will first replace the counter and then replace the ID. The packet will maintain the ID and counter that were removed from the table, and will use those in the next vector in the table (using the new ID for the hash function and comparing counters using the new counter). 

Looking at Fig.~\ref{fig:algorithm overview and flow} and for the sake of the example assume that for each switch, \({G-TopK}\) starts as a copy of \({Sum}\) (meaning the values from its own \({Sum}\) table were processed first and populated \({G-TopK}\) and only then it started processing packets from other switches),
we can see that in switch $1$ in the second slot (which $f_4$ and $f_2$ are hashed into), because the counter of flow $f_4$ is larger, it replaces the values of the ID and the counter, so switch $1$ holds ($f_4, 600$) in its ID and counter values.  Note that if the number of vectors $d > 1$ (unlike Fig. \ref{fig:algorithm overview and flow} in which $d = 1$), the packet will continue to the next vector  of the table with the values ($f_2, 500$) (this also means it uses $f_2$ for the hash table in the next vector). Similarly in switch $2$, flow $f_5$ replaces the stored $f_3$ in the third slot they are both hashed to, and the packet continues processing in switch $2$ with ($f_3, 100$).
If the packet's counter is smaller than the viewed counter, it will not change the counter nor the ID, which we can see in Fig.~\ref{fig:algorithm overview and flow} where in switch $1$, $f_3$ does not change $f_5$ and in switch $2$, $f_2$ does not change $f_6$.

In order to get identical global top-k in every switch in the network, \sysName does an additional comparison of the IDs if the counters are identical (we explain this
requirement in section~\ref{sec:networkWide}). If the counters are \emph{identical} 
\sysName will compare the packet's ID with the saved ID. If the packet's ID is larger, it will switch the IDs. 
If the IDs are identical, the packet will stop processing to avoid replacing other counters in the table in different vectors and creating multiple copies of the same heavy flow. As we can see in Fig.~\ref{fig:algorithm overview and flow} both switches hold $f_1$ in their respective \({Sum}\) tables, so we will receive a copy of  $f_1$ (with the same global counter) from each of these switches.  
At the end of the Consolidation round 
 each switch holds the global top-k heavy flows in \({G-TopK}\), in section~\ref{sec:networkWide} we prove that at the end of the Consolidation round all switches hold an identical global top-k flows table.

We note that even though we filter low frequency IDs in \({G-TopK}\), we filter them according to the locations they are hashed to. 
In Fig.~\ref{fig:algorithm overview and flow}, even though $f_2$ has a larger counter compared to $f_5$, they are hashed to different locations and are compared to different flows, which results in $f_5$ remaining in the table in the example while $f_2$ is filtered. However heavy flows are less likely to be filtered with large enough tables with additional vectors.

Lastly, \sysName needs to maintain a table of the Global top-k that may be queried as needed. Thus, we need a fifth table - \({Query}\), which is a snapshot of the \({G-TopK}\) table at the end of a complete \sysName cycle.
 Once \sysName finishes one iteration of creating a global top-k table, it starts another iteration to keep the tables up to date.
 Recall that the information sharing process is based on Swish, and it can identify when an information sharing round ended, at which time it starts the next round. \sysName uses the same process and similarly uses it to decide when to switch from the Aggregation round to the Consolidation round, as well as after finishing the Consolidation round and starting the next iteration of \sysName (Aggregation round of the next iteration).

%% file: networkwide.tex
\section{Identical Global Top-K Tables}
\label{sec:networkWide}

In this section we prove why \sysName guarantees that every switch eventually holds an identical \({G-TopK}\) table. We note that while network-wide identical \({G-TopK}\) tables are not required for finding the global top-k flows, it is a useful guarantee since \sysName competes with algorithms that use a controller, which will get identical results in all switches, and by guaranteeing uniform results network wide, \sysName does not fall short of a controller based solution in that aspect.
We assume the local top-k algorithm creates a \({L-TopK}\) table without duplicate flow IDs (like Precision \cite{PrecisionJournal}). In this section, the only packets we will refer to are \sysName packets, each containing an ({ID},count) 
pair.
Each table is comprised of $d$ vectors (found in different stages). Each vector contains $s$ pairs of IDs and counters, and a corresponding hash table. For example, (\({GTopK^{i,j}.ID}\), \({GTopK^{i,j}.count}\)) are the values saved in \({G-TopK}\) in vector $i$ in index $j$, and \({hash_{i}}\) is the hash table for the i-th vector. 

\noindent Note that \({hash_{i}({GTopK^{i,j}.ID}) = j}\) (\({GTopK^{i,j}.ID}\) is hashed to index $j$ in the i-th vector).

In order to prove this we will first show that after finishing the Aggregation round, each \({Sum}\) table holds the global counters of the flows in its local \({L-TopK}\) table. This means that in the Consolidation round, for any two packets that are sent between any two switches, if the packets have the same $ID$, they also have an identical $count$.
We will then show that after the end of the Consolidation round, for any (\({GTopK^{i,j}.ID}\), \({GTopK^{i,j}.count}\)) pair, that pair could not have been stored in an earlier vector $k<i$ of \({G-TopK}\). This means that there cannot be any duplicate values in \({G-TopK}\) (two pairs of the same $ID, count$) 
Finally, we show that every switch eventually holds an identical \({G-TopK}\) table.

\begin{lemma}\label{lemma:sameSum}
Under the assumption that each packet will be sent and delivered successfully exactly once.
After finishing the Aggregation round, every switch that holds a specific \textit{ID} in its \textit{Sum} table, will hold the same \textit{counter} for that \textit{ID}.
\end{lemma}
\begin{proof}

Every switch will receive each and every pair of \newline (\({LTopK^{i,j}.ID}\), \({LTopK^{i,j}.count}\)) saved in the network exactly once, since each \sysName packet will be sent and delivered successfully exactly once.
This means that for a given $ID$, every switch that holds it in \({Sum}\) will aggregate all of the counters for $ID$ network wide.
Since each of these switches sums up the same counters and uses every counter exactly once the resulting \({Sum^{i,j}.count}\) for \({Sum^{i,j}.ID}\) will be equal for every switch that holds that same ID.
\end{proof}

From lemma~\ref{lemma:sameSum} we learn that in the Consolidation round if a switch receives multiple packets with the same ID, they will all have the exact same global counter.
This also means that when comparing counters, and since we compare IDs if we find matching counters (as described in section~\ref{sec:system}), we will get an equality if and only if the IDs match.

Because of that, and for simplicity, in the following proofs we only refer to comparisons between counters, and assume that equal counters means they also have identical IDs (we consider both ID and counter as one large 'counter' for the comparisons).

We will now use the following lemma to prove that after the Consolidation round ends, each ($GTopK^{i,j}.ID$, $GTopK^{i,j}.count$) cannot be placed in an earlier vector $k<i$.
\begin{lemma}\label{lemma:orderedCounters}
For any (\({GTopK^{i,j}.ID}\), \({GTopK^{i,j}.count}\)), and for any index $j'$ that \({GTopK^{i,j}.ID}\) would have hashed into in the n-th vector, \({GTopK^{i,j}.count}\) would be smaller than \({GTopK^{n,j'}.count}\).
\[\forall 1\leq i\leq d, 1\leq j\leq s,  \forall 1\leq n<i\]
\[j' = hash_n(GTopK^{i,j}.ID)\]
\[GTopK^{i,j}.count < GTopK^{n,j'}.count\]
\end{lemma}
\begin{proof}

The intuition for the proof is when a packet saves a pair of (ID, counter) in the table in some vector $i$, the packet either entered the pipeline with those values and every value it compared to was larger than the packet's counter until vector $i$, or the packet got those values from switching its starting values with the values of an earlier vector $j$, and then its counter was still smaller than counters between vectors $j$ and $i$. As for values before the vector $j$, we reach the same conclusion of how the values reached vector $j$ in the first place, which had to have been through a packet that either started with those values or switched them in an even earlier stage, and we go on until we reach the first vector whose values could only come from packets that entered the pipeline with those values.  \newline
We will prove using induction. 
We will start by proving our claim for the pair (\({GTopK^{2,j}.ID}\), \({GTopK^{2,j}.count}\)) in the second vector for some arbitrary $1\leq j\leq s$. For convenience we will refer to the pair as ($f_{ID}$, $f_{count}$). In order to be saved in the second vector and since \({G-TopK}\) starts empty, ($f_{ID}$, $f_{count}$) had to have been on a packet processed in \({G-TopK}\), and when comparing $f_{count}$ with the saved counter, $f_{count}$ was larger and replaced the value there. Before that there are two options for how ($f_{ID}$, $f_{count}$) became the packet values:
1) The packet's original data (when it was inserted to \({G-TopK}\)) was ($f_{ID}$, $f_{count}$). Since the packet's data did not change till the 2nd vector, it means that in the 1st vector, in the index \({j'=hash_1(f_{ID})}\) the packet was hashed to, \({GTopK^{1,j'}.count} > f_{ID}\).
2) the packet's original values were a different pair (\(Packet.ID\), \(Packet.count\)) while ($f_{ID}$, $f_{count}$) was saved in the 1st vector, and because \(Packet.count>f_{count}\) the values were switched out and when moving on to the second vector the packets values were ($f_{ID}$, $f_{count}$).
In both cases \({GTopK^{2,j}.count} = f_{count} < {GTopK^{1,j'}.count}\), meaning \({GTopK^{2,j}.count}\) ends up smaller than the counter in the 1st vector in the index \({GTopK^{2,j}.ID}\) is hashed to, note that even if the counter in the first vector is replaced later on it can only happen if a larger counter takes its place, which still makes it larger than \({GTopK^{2,j}.count}\).
We assume correctness for the k-th vector: 
\[\forall 1\leq i\leq k, 1\leq j\leq s,  \forall 1\leq n<i\]
\[j' = hash_n(GTopK^{i,j}.ID)\]
\[GTopK^{i,j}.count < GTopK^{n,j'}.count\]
We now look at the pair (\({GTopK^{k+1,j}.ID}\), \({GTopK^{k+1,j}.count}\)) for some arbitrary $1\leq j\leq s$. For convenience we will refer to the pair as ($f_{ID}$, $f_{count}$). In order to be saved in the (k+1)-Th vector, ($f_{ID}$, $f_{count}$) had to have been on a packet processed in \({G-TopK}\), and when comparing $f_{count}$ with the saved counter, $f_{count}$ was larger and replaced the value there. Before that there are 2 options for how ($f_{ID}$, $f_{count}$) became the packet values:
1) The packet's original data (when it was inserted to \({G-TopK}\)) was ($f_{ID}$, $f_{count}$). Since the packet's data did not change till vector $k+1$, it means that in every vector $1\leq n\leq k$, in every index $j'=hash_n(f_{ID})$ \(f_{ID}\) was hashed to, \(f_{count} < {GTopK^{n,j'}.count}\).
2) The packet's original values were a different pair (\(Packet.ID\), \(Packet.count\)) while ($f_{ID}$, $f_{count}$) was saved in a preceding vector $1<=i<=k$ in the table. When the packet reached vector $i$ it has some (\(Packet.ID'\), \(Packet.count'\)) pair (it can be a different pair than the one it started with), where $Packet.count' > f_{count}$  and the values were switched. The packet then continued from vector $i+1$ with the values ($f_{ID}$, $f_{count}$). Then in every vector $i+1\leq n\leq k$, in every index $j'=hash_n(f_{ID})$, if $f_{count}$ was larger than the saved counter \({GTopK^{n,j'}.count}\), the packet would have switched data again and ($f_{ID}$, $f_{count}$) wouldn't have reached vector $k+1$, similarly if the saved counter (which also includes ID for this comparison) was equal, the packet would have turned inactive and it would not reach vector $k+1$. This means that in every vector from $i+1$ to $k$, in every index $f_{ID}$ was hashed to, the saved counter was larger than $f_{count}$. But from induction we also know the same holds for every vector from $1$ to $i$ since $i<=k$. 
while ($f_{ID}$, $f_{count}$) was on its way to vector $k+1$, counters in previous vectors could change in the meantime as well as after it reached vector $k+1$, but they can only be replaced by larger counters, meaning $f_{count}$ would still be smaller than any counter saved in any index $j'$ that $f_{ID}$ is hashed to in vectors $1$ to $k$.
\end{proof}

An immediate result from Lemma \ref{lemma:orderedCounters}, is that \({G-TopK}\) has no duplicate pairs. The same $(ID, count)$ pair cannot appear twice in the same vector $k$ since they are both hashed to the same index in $k$. and they cannot appear in different vectors $k_1\neq k_2$ since it contradicts Lemma \ref{lemma:orderedCounters}.

\begin{theorem}
At the end of \sysNameNoSp's two rounds of information sharing, the resulting \({G-TopK}\) table will be identical in all switches.
\end{theorem}
\begin{proof}
The intuition for the proof is as follows: We assume the resulting tables are not identical and look at the first different vector $i$ between the two switches. Since the vectors are different, one would have a larger counter in some index. But both switches received the same packets, meaning the switch with the lower counter received a packet with the larger counter, which had to have been placed in an earlier vector $j$ (or it would have replaced the smaller counter in vector $i$ when the packet passed over it).
But in vector $j$ both tables are identical, which means the larger counter is saved in vector $j$ in both switches. This means the switch that has the larger counter in vector $i$, has a duplicate counter in different vectors (since the larger counter is saved in both vectors $i$ and $j$). This contradicts the result of Lemma \ref{lemma:orderedCounters}. \newline
We will prove by Induction.
We start by proving that the 1st vector (of both IDs and counts) is the same in all switches.
If two switches (we will call the first switch $S1$ and the other $S2$) have a different 1st vector - it means one of them (without loss of generality we choose $S1$) has a larger counter in some index $1\leq j\leq s$ in that vector $GTopK_{S1}^{1,j}.count > GTopK_{S2}^{1,j}.count$.
However, if $GTopK_{S1}^{1,j}.count$ was saved in $S1$ it means a packet with that count value also reached $S2$ at some point and that counter had to hash into the same index $j$ in the 1st vector on $S2$ since they use identical hash tables.
If $GTopK_{S1}^{1,j}.count$ is larger - it should have replaced the counter saved there (either while $GTopK_{S2}^{1,j}.count$ was saved there or while a different smaller counter was saved there).
Since it did not replace the saved counter, it means that $GTopK_{S1}^{1,j}.count\leq GTopK_{S2}^{1,j}.count$, which contradicts the assumption $GTopK_{S1}^{1,j}.count >$ $GTopK_{S2}^{1,j}.count$. This proves the 1st vector is identical in every switch in the network.

We assume that the first $k-1$ vectors of \({G-TopK}\) in every switch in the network are identical and will prove it for vector $k$.
We assume the k-th vector is different between two switches ($S1$ and $S2$) - it means one of them (without loss of generality we choose $S1$) has a larger counter in some index $1\leq j\leq s$ in that vector $GTopK_{S1}^{k,j}.count > GTopK_{S2}^{k,j}.count$.
However, if $GTopK_{S1}^{k,j}.count$ was saved in $S1$ it means a packet with that count value also reached $S2$ at some point.
$GTopK_{S1}^{k,j}.count$ couldn't have been placed in an earlier vector $1\leq k' < k$ in $S2$, Since all the previous vectors in $S1$ and $S2$ are identical by induction, so $GTopK_{S1}^{k,j}.count$ would also need to be saved in the same vector $k'$ in $S1$, contradicting the result from Lemma \ref{lemma:orderedCounters} that \({G-TopK}\) has no duplicate pairs. 
This means that $GTopK_{S1}^{k,j}.count$ reaches the k-th vector in $S2$ and it is hashed into the same index $j$ as $S1$ since the hash tables are identical.
If $GTopK_{S1}^{k,j}.count$ is larger than the counter saved there in switch $2$ - it should have replaced it.
if its smaller or equal than the counter there - it contradicts the assumption that $GTopK_{S1}^{k,j}.count$ is the larger value between the two switches in index $j$.
\end{proof}

%% file: Evaluation.tex
\section{Evaluation}\label{sec:evaluation}
We evaluate \sysName for various performance metrics, and compare \sysName to controller based approaches.
We show that \sysName achieves a recall rate of over 95\% of the top-k flows with at most 288KB of memory for \sysNameNoSp's tables.
Furthermore, we implemented \sysName in P4 code for the Intel Tofino Wedge-100 programmable switch~\cite{tofino}, using $\approx$ 2000 lines of code, and show the resource usage of switch resources.

\textbf{System Setup.}
We simulated a network of up to 100 switches, using python and c++ to simulate \sysNameNoSp's operation.  The local top-k algorithm used by \sysName in all of our evaluations is Precision~\cite{PrecisionJournal}, which we have implemented and incorporated into \sysNameNoSp. 
All of \sysNameNoSp's tables were comprised of two vectors ($d=2$) with a varying vector size $s$.
The size of these 
vectors was identical network-wide. The associated hash tables were also identical between different tables in \sysNameNoSp's run (L-top-k table, G-top-k table, etc.).
We chose to use $d=2$ vectors for each of \sysNameNoSp's tables due to results in \cite{PrecisionJournal} showing that Precision does well with only 2 vectors, in addition to memory concerns. We also chose to set $K=128$ in all tests, and vary other parameters.
In each evaluation, we ran each simulation at least 5 times and present here the average results.

\textbf{Datasets.} We used two types of traces: 1) Synthesized traces with different Zipfian distributions, each containing 100M packets and 10M unique flows. The distributions used to generate the traces were $a$=0.6, 0.8 and 1.0. 2) A trace of real traffic, namely the CAIDA UCSD Anonymized Internet Traces - 2018, 2019 \cite{CAIDA}.
In order to simulate large networks 
we merged consecutive CAIDA traces to create larger traces.

\textbf{Splitting the Stream.} In order to challenge \sysName we split the stream between the switches in the following fashion: the top 128 flows were split uniformly across all switches, meaning whenever a packet that belongs to the top-k flows arrives it can go to any switch with an equal chance. The rest of the flows have dedicated switches, meaning all their packets reach the same switch.
Such a split is harder for \sysName since it  would increase the chances of missing the larger flows in some of the local top-k tables, making it harder for \sysName to identify them as heavy flows. 
We note that we have tested additional flow affinities (e.g., 50\% or 80\% flow affinity to a certain switch), and found that \sysName behaves better with higher affinity as the local top-k tables are able to provide a more accurate estimate of the flow counts and the aggregation process is less significant. 

\subsection{\sysName Performance}
Our key observation is that larger network sizes generally have worse results than smaller network sizes, however, given a large enough table size (larger $s$) \sysName performs well. That said, in some cases the larger network actually performs better, we suspect that 
this is due to the fact that the overall saved information across the entire network is larger, 
enabling \sysName to monitor more flows and thus get better results. 

\begin{figure*}[!t]
\centering
    \begin{subfigure}{0.495\textwidth}
    \includegraphics[width=1.0\linewidth]{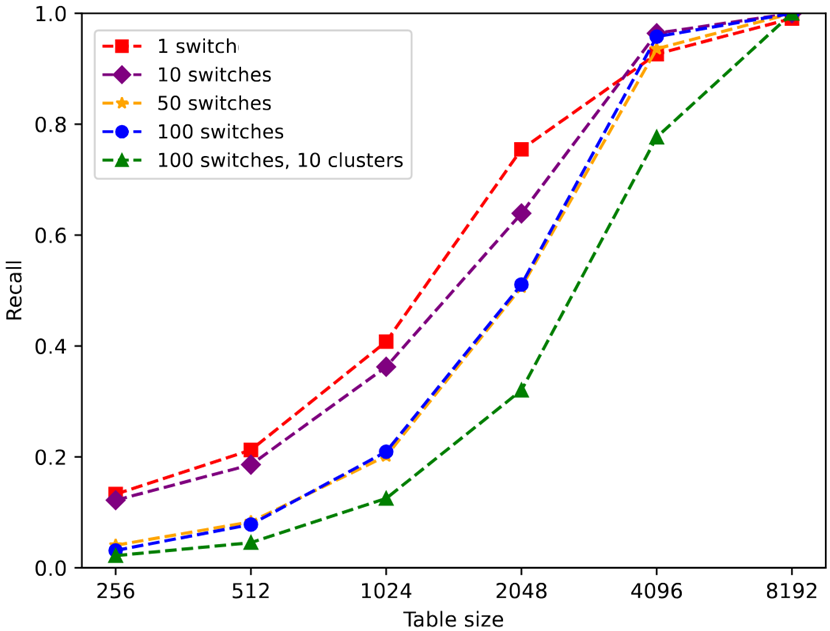} 
    \caption{Synth Zipf 0.6}
    \label{fig:synth_0.6_results}
    \end{subfigure}
    \begin{subfigure}{0.495\textwidth}
    \includegraphics[width=1.0\linewidth]{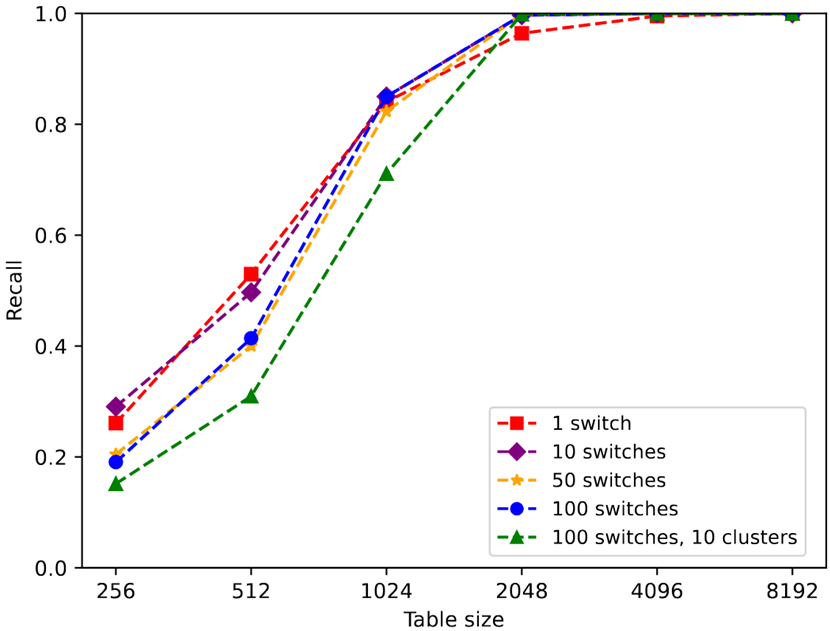} 
    \caption{Synth Zipf 0.8}
    \label{fig:synth_0.8_results}
    \end{subfigure}
    \begin{subfigure}{0.495\textwidth}
    \includegraphics[width=1.0\linewidth]{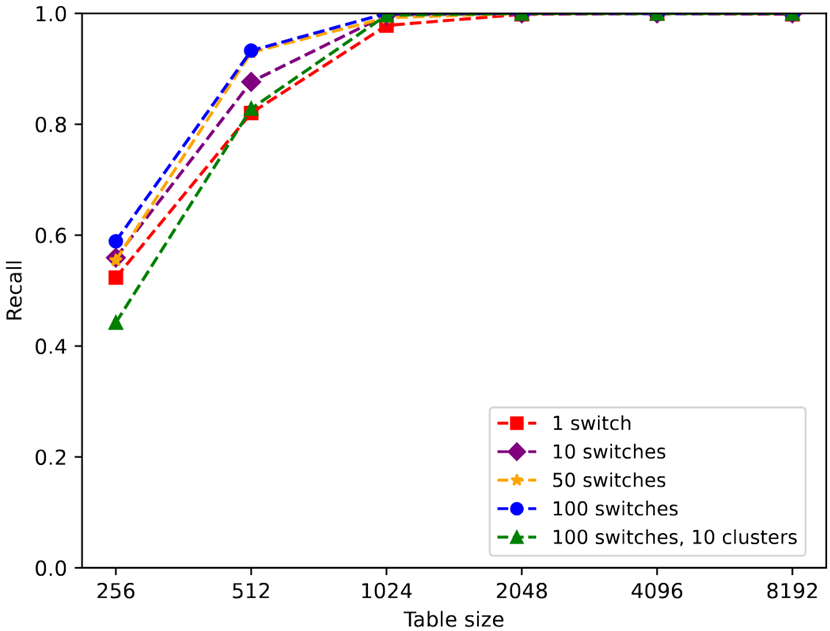} 
    \caption{Synth Zipf 1.0}
    \label{fig:synth_1.0_results}
    \end{subfigure}
    \begin{subfigure}{0.495\textwidth}
    \includegraphics[width=1.0\linewidth]{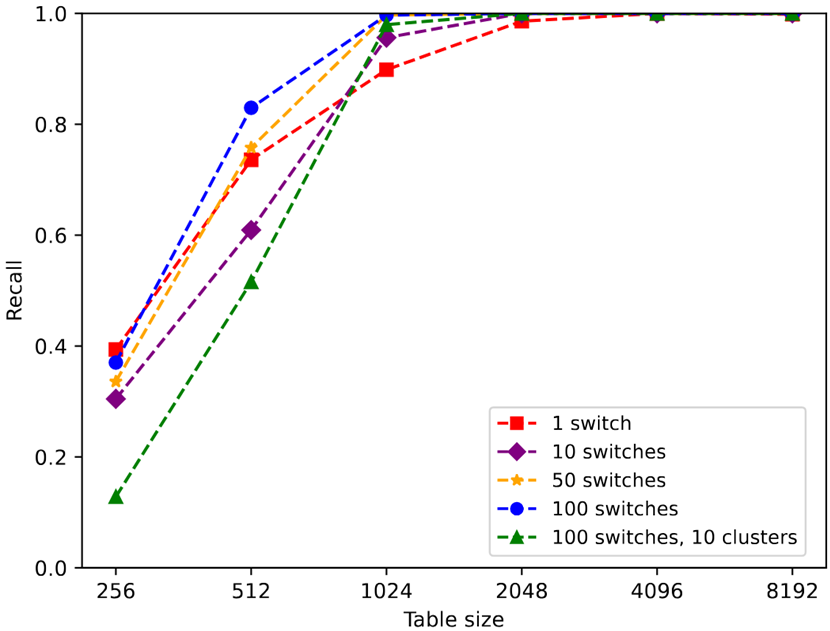}
    \caption{CAIDA}
    \label{fig:caida_results}
    \end{subfigure}
    \begin{subfigure}{0.5\textwidth}
    \includegraphics[width=1.0\linewidth]{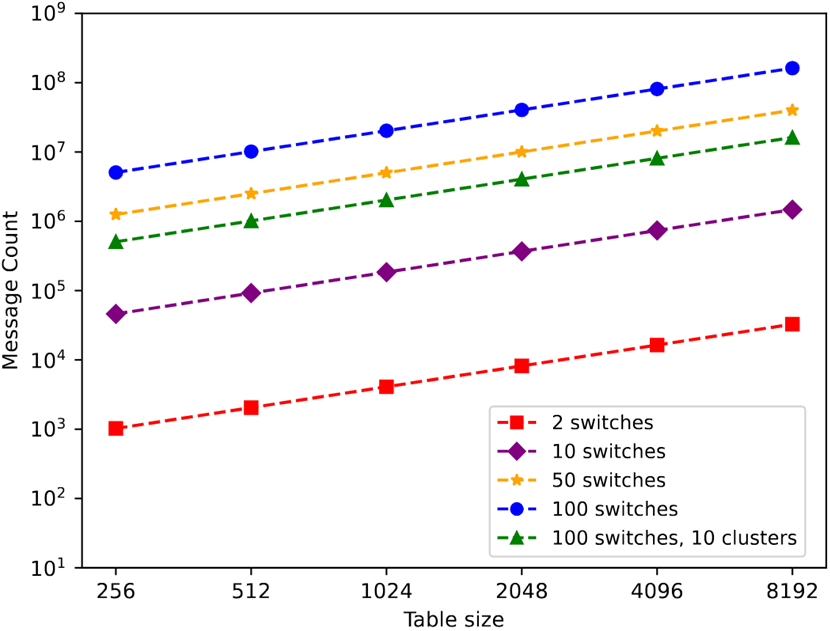} 
    \caption{Message Counts}
    \label{fig:message counts}
    \end{subfigure}

\caption{\sysNameNoSp's recall with different datasets.}
\label{fig:synth and caida results}
\end{figure*}

Fig.~\ref{fig:synth and caida results} shows results on various traces. Note that in results on a single switch, \sysName behaves like Precision~\cite{PrecisionJournal}, since there is no network-wide information to merge.
In comparison, even with a lot more switches, despite harsher flow splits, \sysName performs well. We suspect that is due to how \sysName operates - more switches create a larger effective memory.
Additionally, as we observe results with increasing Zipfian distributions we can see \sysName doesn't need a large table to reach a high recall with Zipfian distribution of 1.0. 

\subsection{Comparing Memory Usage}
We compared \sysName with a basic approach, where each switch maintains a local top-k table, and periodically sends this information to the controller. 
The controller gathers all the information from all of the switches and then computes a network-wide top-k table, which it then sends back to each switch. The controller implementation is not limited in memory or computation and may thus maintain all of the collected information from all of the switches and process it in a non-streaming manner. 
The comparison of how well they identify global top-k is shown in Fig.~\ref{fig:\sysName vs other algorithms}. As shown, \sysName achieves a recall that is very close to that achieved by the controller despite the fact that it functions within the confined resources and processing capabilities of the data plane.

\begin{figure*}[t]
\centering
    \begin{subfigure}{0.495\textwidth}
    \includegraphics[width=1.0\linewidth]{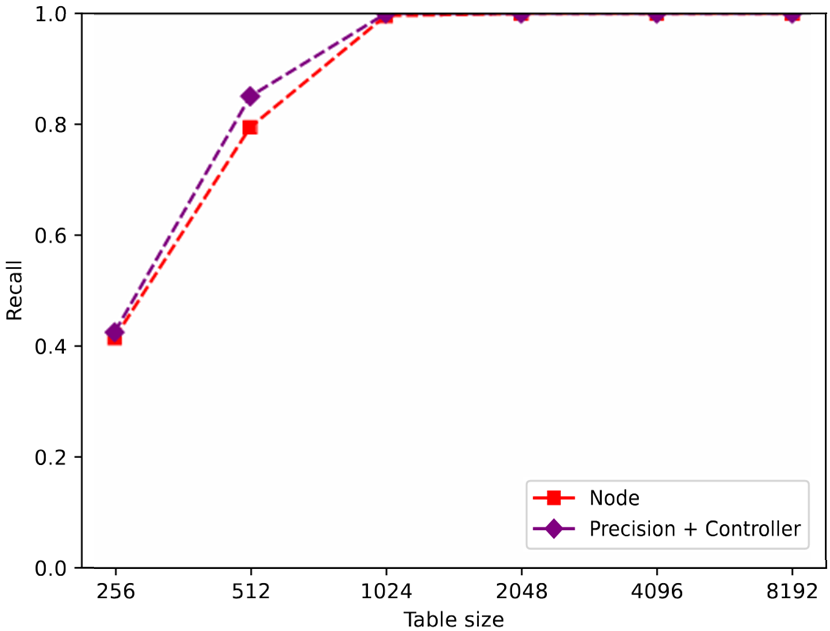} 
    \caption{100 switches}
    \label{fig:compare_different_algo_100_switches}
    \end{subfigure}
    \begin{subfigure}{0.495\textwidth}
    \includegraphics[width=1.0\linewidth]{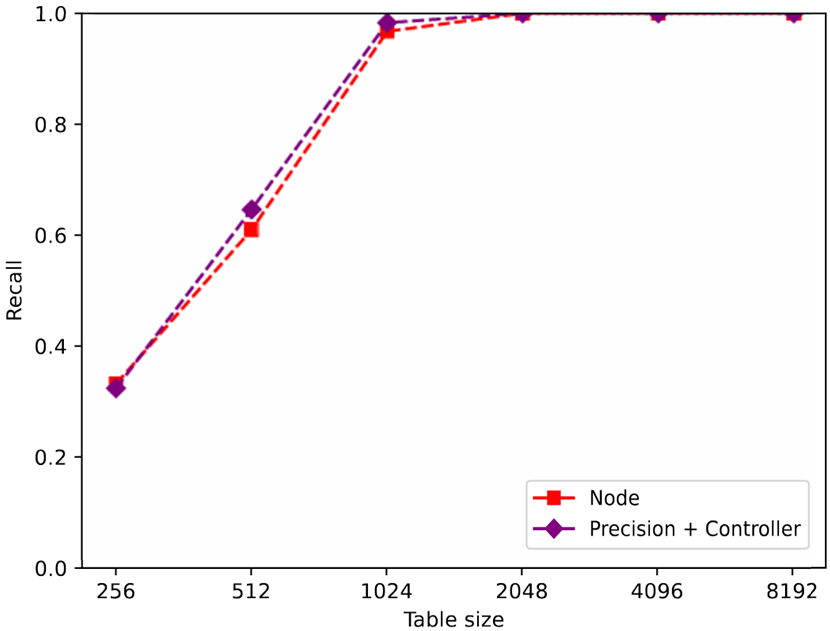} 
    \caption{10 switches}
    \label{fig:compare_different_algo_10_switches}
    \end{subfigure}

\caption{Comparing \sysName to using a local top-k algorithm with a controller.}
\label{fig:\sysName vs other algorithms}
\end{figure*}
We also compare \sysName to the controller based solution of FlowRadar~\cite{flowradar}.
Instead of attempting to emulate FlowRadar and potentially making a weaker version of it,  we compare \sysName to the results shown in the paper~\cite{flowradar}.
FlowRadar 
shows that even with perfect hash tables with no collisions it requires more than 2MB in each switch to support 100k unique flows, 
and over 20MB in each switch to support 1M unique flows.

In our evaluation, the largest table in \sysName uses 8192 cells (two vectors of 4096), with 8-bytes for each cell (4 bytes for ID, 4 for count), totaling 64KB. So 4 tables use 64KB each, and 
\({Sum}\) uses 32KB (since IDs are shared with Snapshot), for a total of 288KB per switch. In the synthetic streams we used 10M different flows and in the merged Caida dataset there are almost 6M different flow IDs. In both, \sysName was able to achieve good results with only 288KB memory in each switch, which is significantly less than that required by FlowRadar.

\subsection{Clustering to reduce communication}
We now look at the communication overhead of \sysName and discuss ways in which this overhead can be reduced. Recall that in each round of \sysName, each switch in the network sends an entire table to every other switch in the network and does so twice. Meaning that as the size of the network grows, the number of messages grows significantly.

\begin{figure*}[!t]
\centering
    \begin{subfigure}{0.495\textwidth}
    \includegraphics[width=1.0\linewidth]{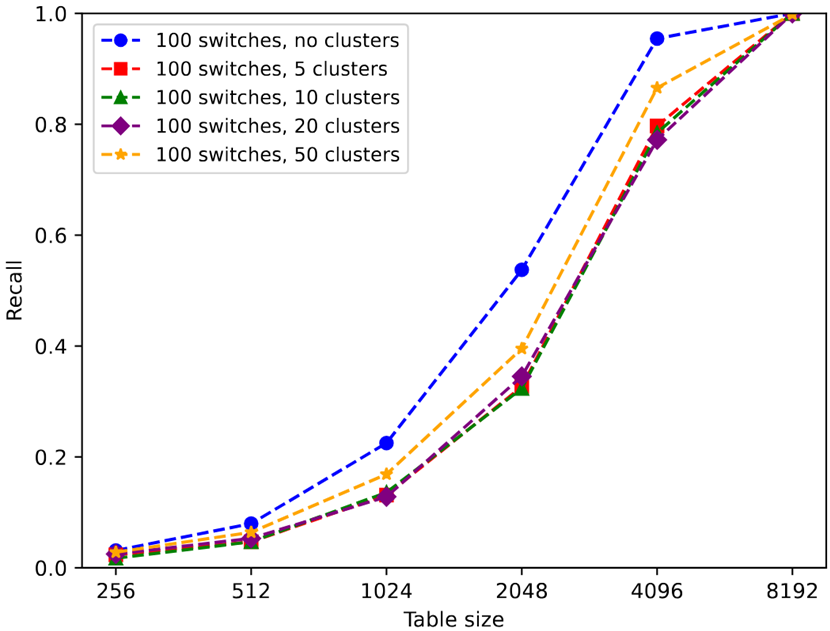} 
    \caption{Synth Zipf 0.6}
    \label{fig:synth_0.6_clustering_results}
    \end{subfigure}
    \begin{subfigure}{0.495\textwidth}
    \includegraphics[width=1.0\linewidth]{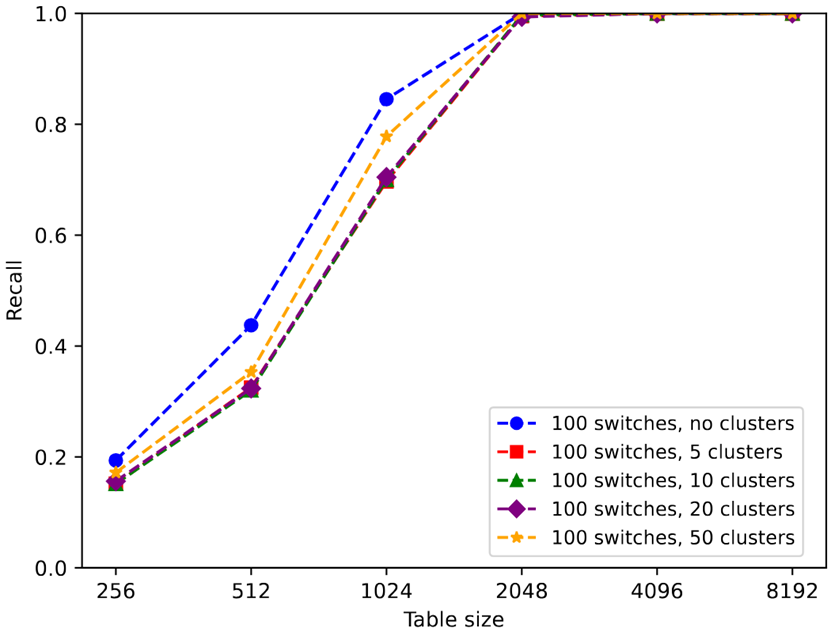} 
    \caption{Synth Zipf 0.8}
    \label{fig:synth_0.8_clustering_results}
    \end{subfigure}
    \begin{subfigure}{0.495\textwidth}
    \includegraphics[width=1.0\linewidth]{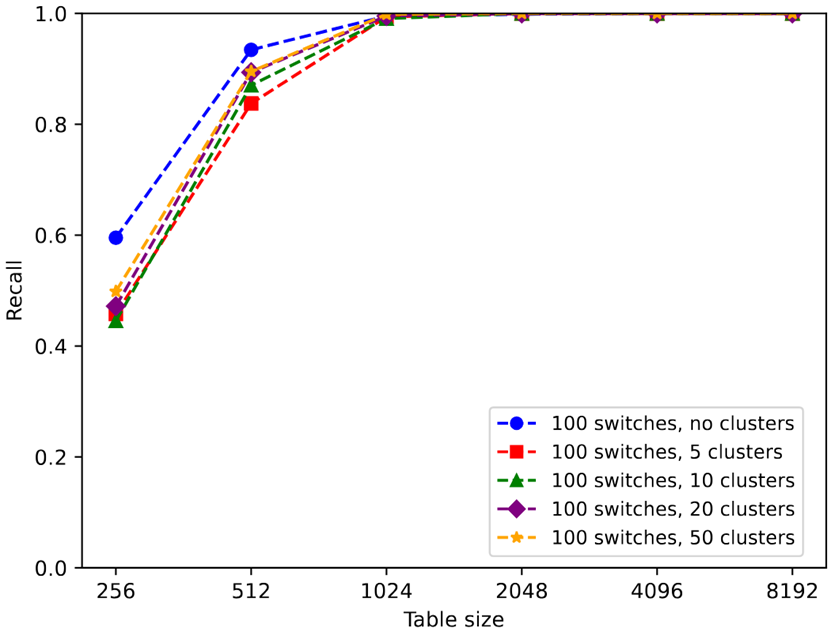} 
    \caption{Synth Zipf 1.0}
    \label{fig:synth_1.0_clustering_results}
    \end{subfigure}
    \begin{subfigure}{0.495\textwidth}
    \includegraphics[width=1.0\linewidth]{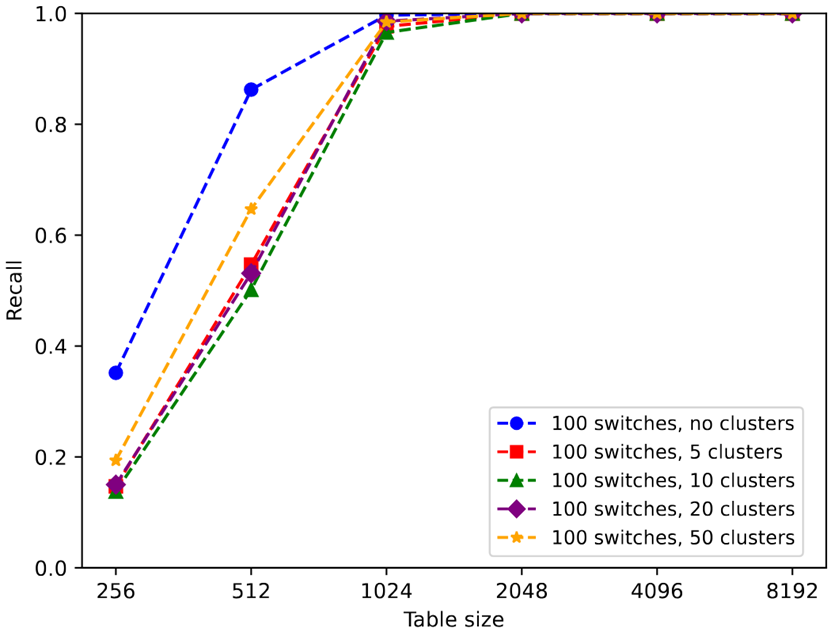}
    \caption{CAIDA}
    \label{fig:caida_clustering_results}
    \end{subfigure}
    \begin{subfigure}{0.495\textwidth}
    \includegraphics[width=1.0\linewidth]{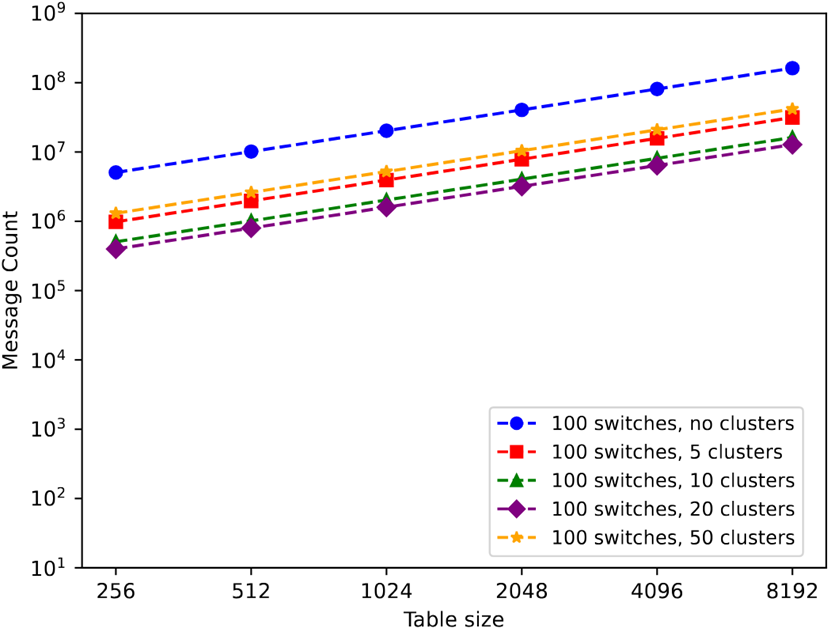}
    \caption{Message Counts}
    \label{fig:message counts clustering}
    \end{subfigure}

\caption{\sysNameNoSp's recall with different datasets when using clustering.}
\label{fig:synth and caida clustering results}
\end{figure*}

As we can see in Fig.~\ref{fig:message counts}, the number of messages increases by \({O(n^2)}\) relative to the network size, a network 10 times larger will have roughly 100~ times more messages.
Specifically in our synthesized test of 100mil packets - in order to get a recall of almost 1 we ended up with more than double the amount of packets, with more than half of the packets being \sysName packets and the rest are regular trace packets. 
Not only does this consume a lot of switch processing time to process all the added packets sent to it by other switches, but this may create substantial communication overhead just from sending the messages depending on the routing between switches.

To address the issue of communication overhead, we designed a simple clustering approach that splits the network into equally (as much as possible) sized clusters.  
Each cluster runs its own independent \sysName algorithm, after which a representative of each cluster will perform an additional round of \sysName with the representatives of the other clusters. 
Each representative's \({G-TopK}\) table acts as the local top-k table for the cluster for the network-wide \sysName. Once the global \sysName finishes, each representative will send the resulting global top-k table to every switch in its cluster.
This means we can reduce the overhead from sending many messages over complex routes, as well as reduce the total number of messages in the network. However, we can expect some drop in performance as well as potentially longer runtime of \sysNameNoSp, since it now has to be run twice (albeit with less messages).
In Fig.~\ref{fig:synth and caida results} we can see 
the results of \sysName with and without clustering. As one might expect, it has lower performance compared to \sysName without clustering, but given a table large enough it competes well and decreases the number of messages in the network by an order of magnitude compared to the same amount of switches without clusters. 

We also compared performances with different amounts of clusters (which also affects cluster size) as shown in Fig.~\ref{fig:synth and caida clustering results}.
As we can see in Fig \ref{fig:message counts clustering}, adding clusters has a large effect on the number of messages circulating in the network between switches for \sysName.

It is important to note that the number of messages required for \sysName does not depend on the amount of traffic that is being processed. The number of messages is a constant (without considering failures) based on the number of switches, clusters and \sysNameNoSp's table size. So while we did double the number of our packets in the synthesized stream. If our trace was 10 times larger - we would have still have used the same number of \sysName messages between the different switches.

\subsection{\sysNameNoSp's Resource Usage}
We implemented \sysName in P4 code for the Intel Tofino Wedge-100 programmable switch~\cite{tofino}.
Our P4 prototype uses $d=2$ vectors in each table with 8192 cells each ($s = 4096$ cells in each vector) to match the simulation's variables. The main resource usage of \sysName is as follows:
    The implementation used all 12 stages available in the switch, that is due to the multiple tables used. The tables also need to be placed in a specific order to enable the functionality of \sysName (i.e Snapshot table in later stages then Local Top-k table). 
Similarly, ALUs are heavily used and Meter ALU usage averages at 64.6\%. 
SRAM usage averaged at 12.1\%, it increases to 28.1\% if we were to use 8 times more cells in each vector ($s = 32768$).  
    Hash Distance Unit averaged at 33.3\%. Note that the same hash functions are used across all of the tables, therefore limiting the amount of hash units needed. 
    
%
Overall, as can be seen Meter ALUs are the most heavily used, yet for other switch resources no more than 34\% is used. It is important to note that this resource usage includes all of the \sysName functionality, including any relevant components of Precision or Swish.  

%% file: Conclusion.tex
\section{Conclusion}\label{sec:conclusion}
In conclusion, we present \sysNameNoSp, an algorithm for efficiently finding the network-wide top-k flows completely in the data plane. 
We proved the correctness of the system analytically and showed its performance through simulation experiments.
We also present a clustering option to reduce \sysNameNoSp's traffic overhead.
In the future we plan to study how clustering can further improve performance and perform evaluation on additional distributed hardware to show the performance improvement achieved by \sysName compared to the centralized approach.\newline
We also note that \sysName can be generalized for any task that faces similar restrictions:
1) uses programmable switches and must adhere to their restrictions which are: feed forward pipeline, limited memory and memory access per pipeline stage and cannot use memory accesses that depend on one another in the same pipeline stage. 2) has two values (ID and counter values are likely, like our own scenario) that require updating (incrementing/decrementing or rewriting/evicting), but different update types require a different dependency order of those values (increment requires ID and then counter while rewriting requires counter and then ID) which cannot be done in the same pipeline pass, but using recirculation cannot help as large changes can occur during packet recirculation.

%% file: Acknowledgments.tex
\section*{Acknowledgments}

\noindent Supported by the Israel Science Foundation no. 980/21.